\documentclass[twoside,leqno]{article}

\usepackage[letterpaper]{geometry}
\usepackage{siamproceedings}

\usepackage[T1]{fontenc}
\usepackage{amsfonts}
\usepackage{physics}
\usepackage{graphicx}
\usepackage{mathtools}
\usepackage{amssymb}
\usepackage{bm}

\usepackage{float}
\usepackage{hyperref}
\usepackage{algorithm}
\usepackage{algorithmicx}
\usepackage{algpseudocode}
\newtheorem{observation}{Observation}
\DeclareMathOperator{\OPT}{OPT}

\title{Improved Online Algorithms for Inventory Management Problems with Holding and Delay Costs: \\
Riding the Wave Makes Things Simpler, Stronger, \& More General}
    \author{David Shmoys\thanks{Cornell University (\email{david.shmoys@cornell.edu}, \email{vs478@cornell.edu})}
    \and Varun Suriyanarayana\footnotemark[1]    
    \and Seeun William Umboh\thanks{The University of Melbourne, ARC Training Centre in Optimisation Technologies, Integrated Methodologies, and
 Applications (OPTIMA) 
  (\email{william.umboh@unimelb.edu.au}). Supported by the Australian Government through the Australian Research Council DP240101353.}}

\begin{document}
\date{}
\maketitle

\setcounter{page}{0}
\begin{abstract}
The Joint Replenishment Problem (JRP) is a classical inventory management problem, that aims to model the trade-off between coordinating orders for multiple commodities (and their cost) with holding costs incurred by meeting demand in advance. Recently, Moseley, Niaparast and Ravi introduced a natural online generalization of the JRP in which inventory corresponding to demands may be replenished late, for a delay cost, or early, in which case there is a holding cost associated with storing it until the desired service time. They established that when the holding and delay costs are monotone and uniform across demands, there is a $30$-competitive algorithm that employs a greedy strategy and a dual-fitting based analysis; notably, they left relaxing the uniformity assumption as an open problem. This assumption is a significant limitation, and in fact, remarkable from the perspective that most online problems with only delay costs do not require uniformity, only monotonicity. 

We develop a $5$-competitive algorithm that handles arbitrary monotone demand-specific holding and delay cost functions, thus simultaneously improving upon the competitive ratio \emph{and} relaxing the uniformity assumption. Our primal-dual algorithm is in the sprit of the work Buchbinder, Kimbrel, Levi, Makarychev, and Sviridenko, which maintains a wavefront dual solution to decide when to place an order and which items to order. The main twist is in deciding which requests to serve early.
In contrast to the work of Moseley et al., which ranks early requests in ascending order of desired service time and serves them until their total holding cost matches the ordering cost incurred for that item, we extend to the non-uniform case by instead ranking in ascending order of when the delay cost of a demand would reach its current holding cost. An important special case of the JRP is the single-item lot-sizing problem. Here, Moseley et al. gave a 3-competitive algorithm when the holding and delay costs are uniform across demands. We provide a new algorithm for which the competitive ratio is $\phi+1 \approx 2.681$, where $\phi$ is the golden ratio, which again holds for arbitrary monotone holding-delay costs.
\end{abstract}

\section{Introduction}
In many problem settings, we are given a collection of clients/requests and we seek to find the most cost-efficient manner in which to serve these requests.
In practice, we often receive these requests over time. Moreover, each client has their own desired time at which point service should occur.
Since it is often impractical
to serve all requests exactly at their desired time, we often relax this
constraint. 
For example, we may allow clients to be served starting after their desired time, where there is a client-specific delay cost function to capture the cost of doing so. This model has been studied in a variety of online settings including, for example, metric service \cite{AzarGGP21}, matchings \cite{Emek16}, latency TSP \cite{BlumCCPRS94}, and facility location \cite{AzarT19}.
 
On the other hand, there are many settings where a late service is impossible, but instead we have a client-specific earliness cost function. This is a critical consideration in inventory management, where the earliness cost functions can capture phenomena such as storage/holding cost for various commodities.

The Joint Replenishment Problem (JRP) is a classical problem from inventory management captured by this framework.
In the offline JRP,
there is a collection of demands, each with their own desired service time and holding (earliness) cost function, as well as a specified demand type. We serve these demands by placing replenishment orders that have an
associated fixed cost for placing any order whatsoever, as well as a fixed cost for each type included in the order (independent of the positive quantity that is ordered). The offline JRP is known to be APX-Hard \cite{NonnerS09, BienkowskiBCDNSSY15}  but admits relatively strong LP-based approximations \cite{BienkowskiBCJNS14}; analogously strong off-line results hold when demands have both holding and delay costs. In contrast, in the online setting, whereas a 3-competitive algorithm is known \cite{BuchbinderKLMS13} when there are only delay costs, the combination of holding and delay costs appears to make the problem significantly more challenging. The best result known previously is a 30-competitive algorithm due to Moseley, Niaparast and Ravi \cite{MoseleyN025}, and this result is restricted to the special case where the delay and holding cost functions are uniform across all demands. Moreover, they showed that their algorithm requires the uniformity assumption: when the delay and holding cost functions are not uniform, then their algorithm's competitive ratio can be linear in the number of requests even for the special case where there is only a single item type, and when the delay and holding cost functions are linear.\footnote{This appears in the arXiv version of \cite{MoseleyN025}  \url{https://arxiv.org/abs/2410.18535}.} Thus, new ideas are needed to remove the uniformity assumption.

Our main result is a 5-competitive algorithm for the JRP with holding and delay costs, even when the associated cost functions are demand-specific, thereby overcoming the restriction of \cite{MoseleyN025} to uniform cost functions. The only restriction on the cost functions that we impose is monotonicity; it is important to note that without the monotonicity assumption, the problem generalizes the set-cover problem, and hence, even in the offline setting, only logarithmic performance guarantees can be achieved (unless $\mathcal {P} = \mathcal {NP}$).
It is important to note that the uniformity assumption made by \cite{MoseleyN025} is a significant limitation. For example, when there are only delay costs, the JRP and its various generalizations (e.g., concave JRP \cite{ChenKU22}, multi level aggregation \cite{BienkowskiBBCDF20,BienkowskiBBCDF21,BuchbinderFNT17}, network design with delay \cite{AzarT19, beyond-tree-embeddings}, and set cover/aggregation with delay \cite{DBLP:conf/latin/CarrascoPSV18, azar2020set}) do not require this uniformity assumption to be amenable to strong competitive analysis results. 
 In addition, our analysis is much simpler than the intricate analysis employed by \cite{MoseleyN025};  this suggests that our results might extend to the more general inventory management settings mentioned above.

In fact, even for the single-item lot-sizing problem (i.e., the special case of the JRP in which there is only one demand type), the online variant with holding and delay costs is already challenging. This is due to the fact that, when we place an order, while it is clear that we should serve all demands whose desired service time has passed, since their delay cost can only increase, it is unclear which of the demands whose service time has not yet passed we should also serve, since their holding cost can decrease in the future.
We give a $(\phi+1)$-competitive algorithm (where $\phi$ denotes the golden ratio) for the single-item lot-sizing problem with both holding and delay costs that are arbitrary monotone functions, which improves upon the previous best known ratio of $3$ (which again is limited to uniform holding-delay costs).
This significantly narrows the gap to
the best-known lower bound of $2$;
this lower bound stemmed from previous work on TCP acknowledgment~\cite{DoolyGS98}, which is equivalent to single-item lot-sizing with only delay costs. Although the focus of our work has been in the design of online algorithms, we have leveraged work on offline algorithms to obtain our results. One surprising aspect of our results for the single-item lot-sizing problem with holding and delay costs is that we give a primal-dual optimization algorithm. One might presume that this means that we have given a complete polyhedral description of the feasible region, but the optimality depends on the cost functions being monotonic.

Much is known about both the offline JRP and online JRP with only one of delay costs and holding costs. Bienkowski et al.~\cite{BienkowskiBCDNSSY15} showed that the offline version is APX-Hard even with only holding costs of $0$ or $\infty$; on the postive side,
for arbitrary monotone holding costs, there is a $1.791$-approximation algorithm
\cite{BienkowskiBCJNS14} based on solving the natural LP relaxation. This algorithm and many of its predecessors \cite{LeviRS06, LeviRSS08} work with minimal (if any) loss in approximation factor when there are both holding and delay costs. In the online setting when there are only delay costs, Buchbinder et al.~\cite{BuchbinderKLMS13} derived a 3-competitive primal-dual algorithm inspired by the offline 2-approximation of Levi et al.~\cite{LeviRS06}. The only previously known result in the online case when there are both holding and delay costs is the work of Moseley, Niaparast and Ravi~\cite{MoseleyN025}.

We adapt the online primal-dual algorithm for delay costs by \cite{BuchbinderKLMS13} to obtain our $5$-competitive algorithm, raising the dual variables in a wavefront manner and placing replenishment orders when unserved demands' dual variables are part of a tight constraint.  To handle the difficulty of determining which demands with a desired service time in the future we should add, we rank the demands with desired service times in the future based on when their delay cost would equal their current holding cost and serve those for whom this condition occurs first; this contrasts with the approach of \cite{MoseleyN025} that ranks demands in ascending order of desired service time. The intuition behind our approach is that we only regret not serving a request early if we end up serving it later when its delay cost is higher than the current holding cost.

This idea of considering the time at which a demand's delay cost would equal its current holding cost is also employed by the concurrent, independent work of Azar and Lewkowicz \cite{azarl25}, who call this future timestep the 'virtual deadline' of a demand. However, their approach to determine if/when a replenishment order ought to be placed relies more directly on the delay cost accrued by unserved demands.

\section{Preliminaries}

In this section, we give the precise statements of the decision-making models that will be addressed in this paper. Throughout this paper, we will be making use of the following standardized notation: for real numbers $x,y$, we define $(x)^+ = \max\{0, x\}$ and $\mathbf{1}_{x \geq y}$ to be $1$ if $x \geq y$ and $0$ otherwise. We let $[k]$ denote the set of positive integers that are at most $k$.

\noindent
\paragraph{The Joint Replenishment Problem with Holding and Delay Costs (JRP):}  In the offline Joint Replenishment Problem, we are given a set of item types $[N] = \{1,\ldots,N\}$, timesteps $[T] = \{1,\ldots,T\}$, a general ordering cost $K_0 \geq 0$ and item ordering costs $K_i \geq 0$ for each item type $i \in [N]$. We are also given a set of demand points $D$ where each demand point $d \in D$ is an ordered pair $(i,t) \in D$ where $i \in [N]$ specifies the item type of $d$ and $t \in [T]$ is its desired service time. Since we are not making any assumptions about the nature of the time steps (other than they are ordered), we may assume without loss of generality that, by splitting timesteps into multiple identical copies, there is at most one demand of any item type $i$ with desired service time $t$. We will use the notation $D_i$ to refer to the set of demands of item type $i$. Each demand point $d = (i,t)$ also has a \emph{holding-delay cost function} $H^{it}_s$ which, if $s \leq t$, is the holding cost associated with servicing the demand $(i,t)$ at time $s$, or, if $s>t$, is the delay cost associated with serving the demand $(i,t)$ at $s$; we also sometimes write $H^d_s$ to refer to the holding-delay cost function for a demand $d=(i,t)$. The holding-delay cost function is monotone non-increasing for $s \leq t$ and monotone non-decreasing for $s>t$. In other words, the holding cost does not increase if the demand point is served closer to desired service time; analogously, the delay cost does not increase if it is served closer to the desired service time; of course, in this description, ``closer to'' a desired service time $t$ for holding cost, and ``closer to'' for delay cost still requires those points in time $s$ to be $s\leq t$ and $s> t$, respectively. To simplify notation, we will assume without loss of generality that $H^{it}_t=0$, i.e., that if a demand is served exactly when its requested service time, there is neither a holding, nor a delay cost.

Demand points are serviced by ``replenishment orders". A \emph{replenishment order} (or \emph{order}) $o$ is specified by the timestep $s$ when the order is placed, the set of item types $U$ to be replenished, and the set of demands with item type in $U$ to be served. The cost of the order is $K_0+\sum_{i \in U} K_i$. Note that a replenishment order $o$ is ``uncapacitated" in the sense that the ordering cost incurred does not depend on the number of demands of an item type $i \in U$ that the order $o$ served, and there is no upper bound on the number of demand points that can be served. The holding-delay cost incurred by any demand point $(i,t)$ served by $o$ placed at $s$ is $H^{it}_s$. The goal is to serve every demand point while minimizing the sum of the ordering cost of replenishment orders and the total holding-delay cost incurred by all demand points. We will use the term ``general order" to refer to the placement of a replenishment order and ``item order" to refer to the inclusion of an item type within a replenishment order. We will think of the general order as costing $K_0$ and the item order as costing $K_i$.
In the natural Integer Linear Programming (ILP) formulation of the JRP, the $\{0,1\}$-variable $y_s$ indicates if a replenishment order is placed at time $s$, the $\{0,1\}$-variable $y^i_s$ indicates if an item order for item $i$ is placed at time $s$, \& the $\{0,1\}$-variable $x^{it}_s$ indicates if demand on item $i$ with desired service time $t$ is serviced at time $s$. Let $D$ denote the set of demands.

\begin{align}
\text{minimize} \label{eq:JRP-obj} \quad & \sum_{s=1}^{T} K_0y_s  + \sum_{s=1}^{T} \sum_{i=1}^N K_i y^i_s + \sum_{t=1}^{T} \sum_{s=1}^{t} x^{it}_{s} H^{it}_{s} \\
\text{subject to} \quad & \sum_{s=1}^{T} x^{it}_{s} \geq 1, \quad  && (i,t) \in D; \\
& x^{it}_{s} \leq y^i_s, \quad   && (i,t) \in D,\ s \in [T]; \\
& x^{it}_{s} \leq y_s, \quad   && (i,t) \in D,\ s \in [T]; \\
& x^{it}_{s},\, y^i_s, y_s \in \{0,1\}, \quad  && s \in [T],\ (i,t) \in D, i \in [N]. \label{eq:JRP-bin}
\end{align}
The dual of the Linear Programming (LP) relaxation of (\ref{eq:JRP-obj})--(\ref{eq:JRP-bin}) is as follows.
\begin{align}
\text{maximize} \quad & \sum_{t \in D} b_{it} \\
\text{subject to} \quad & b_{it} -z^{it}_s - z^{it}_{is} \leq H^{it}_s , \quad && (i,t) \in D,\ s \in [T], \\
& \sum_{(i,t) \in D} z^{it}_{s} \leq K_0, \quad && s\in [T], \\
& \sum_{(i,t) \in D} z^{it}_{is} \leq K_i, \quad && s\in [T], i \in [N], \\
& b_t, z^{it}_s, z^{it}_{is} \geq 0, \quad && s \in [T],\ t \in D.
\end{align}

To specify the online version of this problem, initially, the algorithm must respond to time advancing over a discrete time horizon $\{1,\ldots,T\}$; we initiate the process at time $1$ and let time advance. At each discrete moment in time $s$, the algorithm must decide whether to place an order at time $s$, and if it does, which item types to replenish, as well as which demands of these item types to serve. Moreover, initially, we do not know all of the demands. Instead, as we move forward through time, demand points arrive. To simplify notation, we will assume that if a demand $(i,t)$ arrived at a time $a$, then for all $s<a$, $H^{it}_s = \infty$ (i.e., it is impossible to serve a demand before its arrival time). Moreover, we will assume that each demand's arrival time is no later than its desired service time; that is, $a \leq t$. 

\paragraph{The Single-Item Lot-Sizing with Holding and Delay Costs:} The single-item problem is the special case of JRP in which there is only one item type. To simplify notation, we let $K$ be the total cost of placing a replenishment order of this item type. We also abuse notation slightly and use $t$ to refer to both the demand with desired service time $t$ and the timestep $t$.

The single-item lot-sizing problem admits the following Integer Linear Programming (ILP) formulation, in which $y_s$ indicates whether a replenishment order was placed at time $s$ and $x^{t}_s$ indicates whether the demand with desired service time $t$ was serviced at time $s$. Once again, let $D$ denote the set of demands.

\begin{align}
\text{minimize} \label{eq:LS-obj} \quad & \sum_{s=1}^{T} y_s K + \sum_{t=1}^{T} \sum_{s=1}^{t} x^t_{s} H^t_s \\
\text{subject to} \quad & \sum_{s=1}^{T} x^t_{s} \geq 1, \quad  && t \in D, \\
& x^t_{s} \leq y_s, \quad   && t \in D,\ s \in [T], \\
& x^t_{s},\ y_s \in \{0,1\}, \quad  && s \in [T],\ t \in D. \label{eq:LS-bin}
\end{align}
The dual of the LP relaxation of the single-item lot-sizing problem ILP (\ref{eq:LS-obj})--(\ref{eq:LS-bin}) is as follows.
\begin{align}
\text{maximize} \quad & \sum_{t \in D} b_t \\
\text{subject to} \quad & b_t -z^t_s \leq H^t_s , \quad && t \in D,\ s \in [T], \\
\label{eq:single-item-constraint} & \sum_{t \in D} z^t_s \leq K, \quad && s\in [T], \\
& b_t, z^t_s \geq 0, \quad && s \in [T],\ t \in D.
\end{align}

To simplify the arguments that follow for the rest of this paper, we will assume that all holding-delay costs and ordering costs are integer-valued. We will also assume, without loss of generality, by splitting timesteps into multiple copies, that between any pair of consecutive timesteps, the holding-delay cost function changes for exactly one demand $t$.

\section{The Single-Item Lot-Sizing Problem}
The single-item lot-sizing problem in the traditional holding cost
setting has a rich literature in which both complete polyhedral
characterizations \cite{LeviRS06} and dynamic programming solutions have played 
central roles \cite{WagnerW58,Wolsey95}. 

\subsection{An Offline Exact Optimization Algorithm}

We will give a simple primal-dual algorithm that simultaneously produces an optimal integer solution and its matching dual LP optimum.
That is, we will construct a dual solution $(b,z)$ and an integral primal solution with cost $\sum_{t \in D} b_t$. For ease of exposition, we will assume that the time horizon is continuous and the holding-delay cost function for non-integral values of time is the linear interpolation of the values attained at the integer time point $[T]$. We will use the notation $\frac{d H^{d}_s}{ds}(\tau)$ to refer to the right-derivative of the holding-delay cost of demand $d$ at time $\tau$. This algorithm is analogous to the one given by Levi et. al \cite{LeviRS06} that handles only holding costs.

\textbf{The dual solution:} The algorithm that generates the dual solution maintains a {\it wavefront value} $\tau$, 
which is initially set to $1$ and is increased gradually to $T$. The algorithm will also classify each demand as active or inactive. Initially, all demands are active. Once a demand $t$ becomes inactive, its dual variables $b_t$ and $z^t_s$ for all $s \in [T]$ remain frozen (i.e., unchanged) for the remaining duration of the algorithm. The algorithm will strive to ensure that for all active demands $t$ when the wavefront is at $\tau$, $b_t=\mathbf{1_{\tau \geq t}}H^{t}_{\tau}$ and $z^t_s = (b_{t}-H^t_s)^+$ for all $s$. We will sometimes abuse notation by referring to ``the constraint $s$" becoming tight. By this we mean the constraint $\sum_{t \in D} z^t_s \leq K$ becoming tight during the construction of the dual solution. By making the demands inactive in this situation (and hence freezing their corresponding dual variables), we ensure that the dual solution that we construct is feasible.
To achieve this, we will raise $\tau$ at a rate of $1$ and raise $b_{t}$ at a rate of $\frac{d H^{t}_s}{ds}(\tau)$ and for all demands $t$ and times $s$ such that $b_{t} \geq H^{t}_s$, raise $z^t_{s}$ at a rate of $\frac{d H^{t}_s}{ds}(\tau)$ unless for some such $t,s$, $\sum_{d \in D} z^d_s=K$, in which case, we will make $t$ inactive, and freeze $b_t$ and all $z^t_s$ variables corresponding to that demand $t$.

Before describing the primal solution, we introduce some useful notation. For
every constraint $s$ that became tight, let $s_{freeze}$ be the wavefront value at which the constraint $s$ first became tight. For every constraint $s$ that became tight, we define its corresponding interval $I_s$ to be $(s,s_{freeze}]$. For every demand $t$, let its corresponding interval $I^d_t$ be the set of timesteps $s'$ such that $H^t_{s'} \leq b_t$; note that we use the superscript $d$ to clarify that this is the interval corresponding to the demand with desired service time $t$, not the constraint associated with time $t$. 

\paragraph{Description of primal solution:} We will place replenishment orders as follows. Let $O=\emptyset$ and let $S_{tight}$ be the set of timesteps that became tight. We place an order at the largest $s \in S_{tight}$ and add $s$ to $O$. Then, choose the largest $s \in S_{tight}$ such that  $I_s \cap (\bigcup_{s' \in O} I_{s'}) = \emptyset$ and add $s$ to $O$.
(That is, we select an order point $s$ such that its associated interval is disjoint from the intervals covered by orders already included in $O$.) Repeat this process until there no longer exists some $s$ for which $I_s \cap (\bigcup_{s' \in O} I_{s'}) = \emptyset$ exists. We place replenishment orders at times in $O$. For those demands $t$ such that $z^t_s>0$ for some $s \in O$, we serve $t$ at $s$ and call them \emph{primary} demands; and for every other demand $t'$, we serve $t'$ at any $s \in I^d_{t'} \cap O$. Let $s(t)$ be the timestep at which demand $t$ was served.

The following lemma shows that the primal solution serves every non-primary demand.

\begin{lemma}
For every demand $t$, it holds that $I^d_t \cap O \neq \emptyset$.
\end{lemma}

\begin{proof}
Consider a demand $t$ and let $s$ be the constraint that caused $t$ to freeze. By the definition of freezing, we have $b^{end}_t - z^t_s = H_{s}^t$ and so $s \in I^d_t$. Thus, it suffices to  consider the case that $s \notin O$. There exists some $s' \in O$ such that  $s'>s$ and $I_s \cap I_{s'} \neq \emptyset$; otherwise $s'$ would have been added to $O$ in the process described earlier. By definition of the intervals $I_s$ and $I_{s'}$, we have that $s_{freeze}>s'$. Moreover, $b_t$ froze because of constraint $s$ and so we know that $b_t$ must have frozen no earlier than $s_{freeze}>s'$. Therefore, $b^{end}_t>H_{s'}^{t}$ implying that $s' \in I^d_t$.
\end{proof}

Next, we show that for every demand, its dual variable can pay for its holding-delay cost.

\begin{lemma}
\label{single_holding_delay_dual}
For every demand $t$, it holds that $H_{s(t)}^{t} = b^{end}_t - z^t_{s(t)}$ if it is a primary demand, and $H_{s(t)}^{t} \leq b^{end}_t$ otherwise.
\end{lemma}

\begin{proof}
Let $t$ be a demand. Suppose that $z^t_{s(t)} > 0$ and $s(t) \in O$. Then, since $z^t_{s(t)} > 0$, we have that $b^{end}_t - z^t_{s(t)} = H_{s(t)}^{t}$. On the other hand, if $z^t_{s(t)} = 0$, then $s(t) \in I^d_t$ and so $H_{s(t)}^{t} \leq b^{end}_t$ by definition of $I^d_t$.
\end{proof}

Finally, we prove that we can also charge the ordering costs to the dual variables:

\begin{lemma}
\label{single_item_ordering}
For every primary demand $t$, it holds that $z^t_s>0$ for at most one $s \in O$.
\end{lemma}

\begin{proof}
Suppose for contradiction that the statement is false for some demand $t$. Let $s$ be the smallest member of $O$ such that $z^t_s>0$ and therefore $t<s_{freeze}$. We also know that $b_t$ could not possibly have increased after $s_{freeze}$ and so, $b^{end}_t \leq H_{s_{freeze}}^{t}$. Consider some $s' \in O$ such that $s' > s$. We get that $H_{s_{freeze}}^{t} \leq H_{s'}^{t}$ since $I_{s} \cap I_{s'}$ by construction. Thus, we get that $b^{end}_t \leq H_{s'}^{t}$ which implies $z^t_{s'}=0$.
\end{proof}

Let $P$ be the set of primary demands. \cref{single_item_ordering} implies that
the total ordering cost is at most $\sum_{t \in P} z^t_{s(t)}$.
\cref{single_holding_delay_dual} implies that the total holding-delay cost is at
most $\sum_{t \notin P} b^{end}_t + \sum_{t \in P} b^{end}_t - z^t_{s(t)}$.
Therefore, the total cost is at most $\sum_t b^{end}_t$ as desired.

\subsection{Online Algorithm}
The online algorithm will maintain a dual solution in a similar way to the offline algorithm. When the algorithm is at time $\tau$, the wavefront value will also be $\tau$. To construct the primal solution, we choose a set $O$ of timesteps to place replenishment orders such that $I_{s_1} \cap I_{s_2} = \emptyset$ if $s_1,s_2 \in O$. In the online setting, we know that a constraint $s$ became tight only when the wavefront reaches $s_{freeze}$ and have to settle for placing the order at $s_{freeze}$ instead of at $s$. This explains why, for the case when there are only delay costs, one obtains a $2$-approximation. In the holding-delay cost model, now we also have to choose which demands $t>s$ we wish to serve.  Our solution to this problem is a premature service step in which we rank these demands in ascending order of when their dual variables might become large enough to pay for serving them now, i.e., when their delay cost would be equal to their current holding cost. Then, inspired by prior work on rent-or-buy problems, we will choose a set of demands with total holding cost at most $K$, leading to a competitive ratio of $3$, matching the best known guarantee for when the holding-delay cost functions are uniform \cite{MoseleyN025}. As in \cite{MoseleyN025}, we will call a demand $t$ {\it mature} if the wavefront value $\tau \geq t$ and {\it premature} if $\tau < t$ (in which case $b_t=0$). 

\paragraph{High-Level Description:} The online algorithm, which we will call {\sf{ALG-{single}}} will maintain a current dual solution in the same manner as the offline primal-dual algorithm. In particular, the dual variable of an unserved demand increases in tandem with its delay cost. If an unserviced demand $t$ freezes when the wavefront is at $\tau$ we place a replenishment order at $\tau$.  When a replenishment order is placed all unserved mature demands will be served. Next is what we call the premature service step; we will rank the premature demands $t$ in ascending order of when their delay cost would equal the holding cost associated with serving $t$ at time $\tau$.

Sometimes, a demand will be served without freezing its dual variable value. These demands will be referred to as \emph{semi-active}. Sometimes we will abuse terminology and refer to the dual variable of a demand $t$ as active/semi-active. The term unfrozen will be used to refer to both active and semi-active demands.

Conceptually, we are trying to maintain a dual solution in which for every time $s$ and every mature demand $t$ that is unserved at time $s$, we have $b_t=H^{t}_s$. This is desirable since it enables us to bound the delay cost associated with serving any demand $t$ by its dual variable $b_t$. When maintaining this property becomes impossible, we place a replenishment order and serve all demands which are past their desired service time.

To identify which extra demands we should serve at time $s$, we sort the demands with future desired service time in ascending order of when their dual variable would become large enough to freeze because of a constraint corresponding to a time $s' \leq s$. The idea is that either the next replenishment order will be triggered by a constraint corresponding to a timestep $s''>s$, in which case the dual variables must have increased by at least $K$ or all the demands $t>s$ which we served will have dual value at least $H^t_s$ before the next order is placed, again enabling us to bound the total ordering cost. Moreover, the total holding cost of demands served via any single order is at most $K$, implying that the holding costs are smaller than the ordering costs.

\paragraph{The dual solution and replenishment orders:} We will maintain the invariant that when the algorithm is at time $s$, if a demand $t$ is active, then $b_{t}=\mathbf{1}_{s>t} \cdot H^t_s$ and $z^t_q=(b_t-H^t_q)^+$ for all $q \in [T]$. When a demand arrives, it is active and we set $b_{t}=0$. The dual variable $b_{t}$ remains 0 until its desired service time. Suppose that we are at time $s$ and $U$ is the set of unfrozen (i.e., active and semi-active) demands including any demands that may have arrived at time $s$. Let $t$ be the demand whose holding-delay cost changes from $s$ to $s+1$. If $t \in U$ and increasing $b_{t}$ and $z^t_q$ to satisfy the above invariant on the dual variables, that is, setting $b_t$ to $\mathbf{1}_{s+1>t} \cdot H^t_{s+1}$ and $z^t_q$ to $(\mathbf{1}_{s+1>t} \cdot H^t_{s+1}-H^t_{q})^+$, would lead to a constraint being violated, then we do not increase $b_t$ or any $z^t_s$ variable. Instead, the demand $t$ and its corresponding dual variables will be frozen \footnote{This is different to what the dual solution described in the previous solution would do and is technically not required. But it does simplify the competitiveness proofs a little.} and $t$ now becomes inactive and if $t$ was active, then we will place a replenishment order at time $s$. By construction, the only constraints that could be violated are \eqref{eq:single-item-constraint}. Let $s'$ be the timestep corresponding to the constraint that would have become violated. If there are multiple, pick the latest such timestep. We will say that the demand $t$ \emph{triggered} the order at timestep $s$ and also abuse notation slightly to say that the constraint $s'$ triggered the order. We will also make inactive all demands $t' \in U$ such that $t' \leq s$. Once a demand $t$ is frozen, then $b_t$ and $z^t_s$ remain fixed for the remainder of the algorithm.

\begin{figure}
    \centering
    \includegraphics[width=0.5\linewidth]{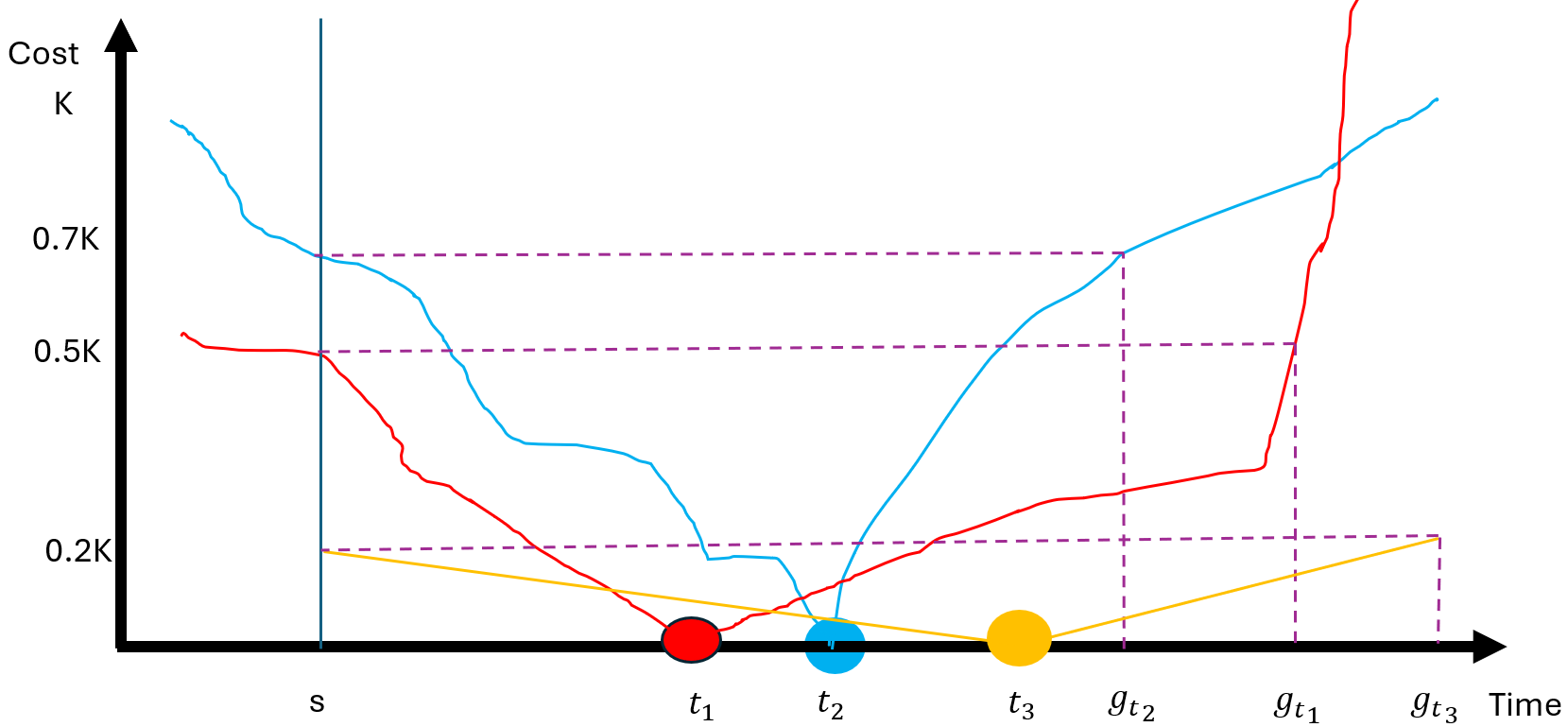}
    \caption{When placing an order at time $s$, we rank demands with desired service time $>s$ in ascending order of when their delay cost equals their current holding cost. For each demand $t_i$, this happens at time $g_{t_i}$. When considering demands in ascending order of their $g_t$ value, we consider demand $t_2$ first because $g_{t_2}<g_{t_1}<g_{t_{3}}$. $H^{t_2}_s=0.75K<K$, so we can serve $t_2$. However, if we also tried to serve $t_1$, the total holding cost would be $1.25K>K$, therefore, we serve $t_2$ but not $t_1$ at time $s$. The algorithm will not even consider serving $t_3$ at time $s$ eventhough $H^{t_2}_s+H^{t_3}_s=0.9K<K$ because $g_{t_3}>g_{t_1}$ and $t_1$ will not be served at time $s$. }
    \label{fig:placeholder}
\end{figure}
\paragraph{Identifying which demands to serve when an order is placed:} Suppose we decide to place an order at time $s$. We will serve all active demands $t \leq s$. Then we perform our premature service step. For any active demand $t>s$, let $g_t \geq t$ be the first time at which $H^t_{g_t} \geq H^t_{s}$. Go through the active demands $t>s$ in ascending order of $g_t$, adding $t$ to the set of demands served by the order as long as the total holding cost of these demands does not exceed $K$. If adding demand $t_i$ serving demand $t_i$ would cause the total holding cost to exceed $K$, then neither $t_i$ nor any demand $t_k$ with $g_{t_k}>g_{t_i}$ is served at time $s$. Let $\tau_s$ be the largest of the $g_t$ values among all demands served in this way. See Figure \ref{fig:placeholder} for an example of this process. In that example, $\tau_s=g_{t_2}$.

We will use the following facts, which follow directly by construction:
(i)
the dual solution $b,z$ is feasible and during the algorithm's execution, each variable never decreases.
(ii) if a demand $d$ is unserved at time $t$, then at all times $\tau<t$, $b_d$ was not part of a tight constraint;
(iii) when the wavefront value was $s$, it holds that $z^d_{s'} = 0$ for all $s'>s$ and all demands $d$;
(iv) the total holding cost incurred by the algorithm is at most its total ordering cost;
(v) the delay cost incurred in serving a demand $d$ is at most $b_d$;
(vi) if $z^t_s$ is increasing at any point in time, then $b^t_s$ is increasing at the same rate.

\begin{lemma}
The ordering costs are at most $\sum_{t \in D} b_{t}$.    
\label{ordering_costs_single_item_3}
\end{lemma}

\begin{proof}
We will prove a slightly stronger claim, the increase in $\sum_{t \in D} b_{t}$ between successive orders is at least $K$. Suppose for contradiction that this is false, i.e., between the $k$th and the $(k+1)$st order, the increase in dual variables was smaller than $K$. Let $s_k$ and $s_{k+1}$ be the times at which the $k$th and $(k+1)$st orders were placed, $t_k$ and $t_{k+1}$ be the demands that triggered these orders, $e_{k}$ and $e_{k+1}$ be the constraints that triggered these orders, and $S$ be the set of demands $t>s_k$ served at time $s_k$. 

If $e_{k+1}>s_k$, then we know that at time $s_k$, $\sum_{t \in D} z^t_{e_{k+1}}=0$. Therefore, if $\sum_{t \in D} z^t_{e_{k+1}}=K$ when the $(k+1)$st order is placed, then $\sum_{t} b_t$ must also have increased by $K$ between $s_k$ and $s_{k+1}$.

If $e_{k+1}<s_k$, then the demand $t_{k+1}$ had arrived before $s_k$, but $t_{k+1}>s_k$ (otherwise $t_{k+1}$ would have been served at $s_k$). This implies that $b_{t_{k+1}} \geq H^{t_{k+1}}_{s_k}$. In this case, we claim that for every demand $t>s_k$ served at $s_k$, $b_t \geq H^t_{s_k}$. Notice that if the claim were true, then $\sum_{t \in S} b_t + b_{t_{k+1}} > K$, otherwise, $t_{k+1}$ would have been served at time $s_k$.

To prove the claim, we know that because $e_{k+1}<s_k$, all constraints (7) associated with the interval $[s_{k}+1,s_{k+1}]$ had not yet become tight at time $s_{k+1}-1$. This implies that for all demands $t' \in S$, either $t'$ is still unfrozen, or it froze because of a constraint corresponding to a time $e \leq s_k$. This implies that $b_{t'} \geq H^{t'}_{s_k}$ for all $t' \in S$.
\end{proof}

\begin{theorem}
{\sf{ALG-{single}}} is a $3$-competitive algorithm.
\end{theorem}

\begin{proof}
The holding cost is at most the ordering cost, which in turn by the previous lemma is at most $\sum_{t \in D} b_t$. Moreover, since the delay cost incurred by each demand $t$ is at most $b_t$, the total delay cost is also at most $\sum_{t \in D} b_t$ giving us a competitive ratio of $3$. 
\end{proof}

\noindent
We show how to improve this bound to $\phi +1$ in Section \ref{sec:phi}.

\section{Extending to the Joint Replenishment Problem}
Our algorithm for JRP essentially combines the algorithm of Buchbinder et al.~\cite{BuchbinderKLMS13} to determine when and which item types to replenish with the concept of a premature backlog phase from Moseley et al.~\cite{MoseleyN025}. However, unlike \cite{MoseleyN025}, this phase will choose demands in order of when their delay cost would first exceed the holding cost of serving them now. In the same spirit as \cite{BuchbinderKLMS13}'s analysis, we will be able to conclude that the total holding cost is at most $2\OPT$ and obtain a competitive ratio of $5$. In this section, we will first describe a slightly simpler algorithm and present a weaker analysis yielding a competitive ratio of $7$. Then we will explain the necessary modifications to the algorithm and present the more careful analysis that improves this ratio to $5$.

\subsection{Algorithm Description}
At a high level, the algorithm will build a dual solution in a similar way as before. When an unserved demand freezes it will place a replenishment order. To determine which extra item types to serve, we will simulate the dual solution assuming no further arrivals for some time and add on the item types corresponding to active demands that froze during simulation. To determine which demands of these item types we serve, our premature service step will use the same ranking method as for single-item lot-sizing. We will be able to show that the solution we construct costs at most $5$-times the dual solution. 

\label{sec:jrp-alg}
We begin by describing how the algorithm maintains the dual solution. This is essentially the same as the previous section. We will refer to this algorithm as {\sf{ALGJRP-{simple}}}.
\paragraph{Dual solution and placing replenishment orders:}  The algorithm maintains a wavefront time $\tau$, which is initially $1$ and will increase gradually to $T$. We will strive to ensure that when the wavefront is at $s$, $b_{it}=\mathbf{1}_{s \geq t}H^{it}_s$ if the demand $(i,t)$ is still unfrozen and $t \leq \tau$ . Initially, all demands are active and we gradually increase the wavefront $\tau$ at a rate of $1$. We also increase the dual variables $b_{it}$ at rate $\frac{\partial H^{it}_s}{ds}$. To ensure that the constraints are satisfied, for all integer timesteps $s$ such that $b_{it}$ we increase $z^{it}_{is}$ at a rate $\frac{\partial H^{it}_s}{ds}$ unless $\sum_{(i,t) \in D_i} z^{it}_{is} = K_i$ in which case we increase $z^{it}_{s}$ at a rate of $\frac{\partial H^{it}_s}{ds}$ if $\sum_{(i,t) \in D} z^{it}_s < K_0$ or, if $\sum_{(i,t) \in D} z^{it}_s = K_0$, then the demand $(i,t)$ and the dual variable $b_{it}$ are frozen. We will say that $(i,t)$ froze because of a constraint at time $s$. If there are multiple such $s$, of which $s'$ is the largest, then we will instead say that $(i,t)$ froze because of a constraint at time $s'$.

Whenever an active demand's dual variable freezes, we place a replenishment order. If the demand $(i,t)$ froze when the wavefront value was at $\tau$ because of a constraint at time $s$, let $S$ be the set of item types such that, when the wavefront value was $\tau$, $\sum_{(i,t) \in D_i} z^{it}_{is}=K_i$. By design of the algorithm, if for some demand $(i',t')$, $z^{i't'}_s>0$, then $i' \in S$. We will say that \textbf{the constraints $S_s$ triggered the order}. By the constraints $S_s$ we are referring to the constraints $\sum_{(i,t) \in D_i} z^{it}_{is} \leq K_i$ (for $i \in S$) and $\sum_{(i,t) \in D} z^{it}_s \leq K_0$. 

\paragraph{Identifying which item types to serve when an order is placed:} Suppose we decide to place an order at wavefront time $\tau$ when a constraint becomes tight. The order will serve every active demand whose dual variable froze at wavefront time $\tau$. If the replenishment order is phase-initiating (to be defined later) these demands will be called regular demands and the item orders the algorithm places to serve them will be called regular item orders. Additionally, for each regular item order of an item type $i$, all mature demands of type $i$ are also served and their dual variables frozen.  Since we have incurred the joint ordering cost, we also want to include other item types and serve other demands if they are close to freezing. To identify which item types we should serve, we will run a simulation that simulates future changes to the dual solution (assuming no new demands arrive). Every unserved demand whose dual variable froze before the simulation ended will also be served by this order. The simulation ends when either 
(a) the increase in dual objective value is exactly $K_0$, or
(b) all dual variables have frozen.
In summary, the item types served by the order are the ones associated with active demands that either froze at time $\tau$ or during the simulation. To clarify, if a demand froze at exactly the moment the simulation ended, the demand is not served, it must freeze before.
 
\paragraph{Identifying which additional requests to serve prematurely:} Having identified which item types we are serving for the order at wavefront time $\tau$, for each item type $i$ that has been included in the order, we need to identify which demands we will serve with this order. Let $i$ be an item type included in the order. As mentioned above, we will serve all unserved demands of item type $i$ with desired service time before $\tau$ as well as all demands of item type $i$ that froze in simulation. Then, for every remaining demand $d$ of item type $i$ that have not yet been included in this order, we let $g^d_\tau$ be the future time at which the delay cost of $d$ would be equal to the holding cost of serving it at $\tau$. Rank these unserved demands $d$ of item type $i$ in ascending order of $g^d_\tau$. Our approach will, similarly to the single item case, be to add on demands of item type $i$ in ascending order of $g^d_\tau$ until the total holding cost would exceed a threshold of $K_i$.

\textbf{Clipping:} If the dual variable $b_d$ of a demand $d$ froze at value $H^d_{\tau_f}$ in a simulation corresponding to an order at time $\tau < \tau_f$, its delay cost function will be clipped: For all $\tau''>\tau_f$, $H^d_{\tau''}:=H^d_{\tau_f}$. Moreover, if a mature demand $i,t$ is served at some time $s \geq t$ but $b_t$ has not yet frozen, we clip its delay cost function in the same way $H^{it}_{\tau} \gets \min (H^{it}_{\tau},H^{it}_s)$ for all $\tau > t$.  

Notice that because every demand is served before its holding-delay cost is clipped, the clipping does not decrease the cost of the algorithm. Therefore, this is without loss of generality.

\textbf{Correct simulations:} For any specific simulation which started at some time $\tau$ and ended at some time $\tau'$, we call the simulation \emph{correct} if, for every demand $d$ that had arrived at or before time $\tau$, at all times $\tau'' \in [\tau,\tau']$, we have that the actual value of $b_d$ at time $\tau''$ is exactly the same value as it did at time $\tau''$ in simulation. Notice that by the construction of the dual solution, this is equivalent to saying that at all times $\tau'' \in [\tau,\tau']$, it holds that the $b_d$ is actually frozen at time $\tau''$ if and only if it was frozen at time $\tau''$ in the simulation associated with time $\tau$. Similarly, we will call a simulation \emph{$i$-correct} if it was correct for all demands of item type $i$.

\textbf{Phase-initiating orders and item orders:} For every replenishment order placed, there is some set of constraints $S_s$ that became tight triggering the replenishment order at time $t$. If there are multiple, we choose the one with latest corresponding $s$. We will refer to $(s,t]$ as the interval associated with this replenishment order. If the interval associated with a replenishment order is disjoint from the intervals of all previous replenishment orders, then it is called \emph{phase-initiating}. Recall that regular item orders are item orders that we decided to place in a phase initiating order before even beginning its corresponding simulation.  Similarly to phase-initiating orders, for every item order, if $S_s$ is the constraint involving an unserved demand of item type $i$ that became tight (possibly in simulation) to force an order of item $i$ at time $t$, we can let $(s,t]$ be the interval corresponding to the item order. We call an item order \emph{item-phase-initiating} if its interval is disjoint from all previous intervals of item orders of item type $i$.
\subsection{$7$-competitive analysis}
We will bound each type of cost (general ordering, item ordering, holding, delay) separately using the objective value of the dual solution $B=\sum_{d \in D} b_d$. Like in the single item case, we will use the increase in dual variables between successive orders to bound the general ordering cost and regular item ordering cost. Because some amount of the dual increase that justifies a non-regular item order at time $\tau$ is in simulation of dual increase after $\tau$, we will need to charge dual increase after the previous item order, and, if the simulation is correct, before the next item order. delay costs will be at most $B$ and holding costs will be at most the total ordering cost. It will turn out that every incorrect simulation is followed by a phase-initiating order, enabling us to bound the item ordering costs even for incorrect simulations.
Facts i-vi (from section 3.2 just before Lemma \ref{ordering_costs_single_item_3}) are true even for this algorithm. Even if for any demand $d=(i,t)$ we replace $z^d_s$ with $z^{it}_{is}$ in the statements of facts iii and vi, it would still be true. These facts enable us to bound the delay and holding costs in a simple manner.

We will use the notation $\Delta^S_{(s_1,s_2]}$ to refer to the change in $\sum_{d \in S} b_d$ as the wavefront value goes from $s_1$ to $s_2$ and
$\Delta^{S}_{s,sim}$ to refer to the change in $\sum_{d \in S} b_d$ during the simulation at $s$. We will let $\Delta^S_{(s_1,s_2],sim} = \Delta^S_{(s_1,s_2]}+ \Delta^{S}_{s,sim}$.

\begin{lemma}
Consider a set of constraints $S_s$ associated with time $s$ that were all tight at wavefront time $s' > s$. Let $S^D$ be the set of demands $(i,t)$ whose dual variables contribute to $S_s$ (i.e. $i \in S$ and either $z^{it}_{is}>0$ or $z^{it}_s>0$) and let $S$ be the set of item types corresponding to $S^D$. Then, the sum of the contributing dual variables $\sum_{d \in S^D} b_d$ must have increased by at least $K_0+\sum_{i \in S} K_i$ between wavefront times $s$ and $s'$.
\label{ordering_after_s}
\end{lemma}

\begin{proof}
When the wavefront value was at $s$, $z^{it}_{is}=z^{it}_s=0$ for all $(i,t) \in D$. When the wavefront was at $s'$, $\sum_{(i,t) \in D} z^{it}_s + \sum_{i \in S} \sum_{t: (i,t) \in D} z^{it}_{is}=K_0+\sum_{i \in S} K_i$, and every increase of either a $z^{it}_{is}$ variable or a $z^{it}_{s}$ variable corresponds to an equivalent increase in $b_{it}$, the result follows.
\end{proof}

\begin{corollary}
If a phase-initiating order is placed at time $\tau$ with interval $(s,\tau ]$ and the corresponding constraint $S_s$ had a corresponding set of item types $S$, then the $\Delta_{(s,\tau]} \geq K_0+\sum_{i \in S} K_i$.
\label{pio_cost_interval}
\end{corollary}

\begin{observation}
If a simulation associated with an order at time $\tau$ ended because all dual variables have frozen, then the next order (if one exists) is phase initiating.
\label{no_early_termination_simulation}
\end{observation}

\begin{proof}
The next order must be triggered by a demand that hasn't arrived yet. Therefore, the next order's interval must begin at a time after $\tau$, making it phase-initiating.
\end{proof}

\begin{lemma}
In a simulation initiated at time $\tau$ that ended at time $\tau'$, the only constraints $s$ that become tight strictly before $\tau'$  have $s<\tau$.
\label{simulation_tight_only_past}
\end{lemma}

\begin{proof}
For any time $s>\tau$, $z^{it}_s=0$ for all $(i,t) \in D$. Therefore, the total increase in the dual objective value must be at least $K_0$ before $\sum_{(i,t) \in D} z^{it}_s=K_0$. Since the simulation ends the moment the total dual increase is $K_0$, we have the desired result.  
\end{proof}

\begin{lemma}
If an order is placed at time $\tau$, the corresponding simulation is correct and the next order is placed at time $\tau_2$, then the total increase in dual variables between $\tau$ and $\tau_2$, $\Delta_{(\tau,\tau_2]}$, is at least $K_0$.
\label{NPIO_general_bounded}
\end{lemma}

\begin{proof}
If the order at $\tau_2$ was phase-initiating, then \cref{pio_cost_interval} implies the desired result. Otherwise, by the contrapositive of \cref{no_early_termination_simulation}, the simulation must have ended because the total increase of dual objective value in the simulation was $K_0$. The simulation must have ended at some wavefront value $\tau_{sim} \leq \tau_2$ (otherwise the active demand(s) that froze before $\tau_2$ would have been served due to the simulation at time $\tau$). Therefore, $\Delta_{(\tau,\tau_2]} \geq \Delta_{(\tau,\tau_{sim}]} \geq K_0$. 
\end{proof}

The following lemma will only be leveraged in the $5$-competitive analysis we go through later, but is nevertheless conceptually illuminating.

\begin{lemma}
    A simulation is correct if the next order is non-phase initiating.
\label{pio_incorrect_simulation}
\end{lemma}

\begin{proof}
If an order is placed at time $\tau$ and the next order (placed at $\tau_2$) is non-phase initiating, then that means the interval associated with the next order has left endpoint before $\tau$. This implies that at time $\tau$, all dual variables contributing to this constraint were already known. Suppose for contradiction that the simulation is incorrect. Let $\overline{\tau} \in (\tau,\tau_2)$ be the first time at which some demand $d$ was frozen in exactly one of the simulation at $\tau$ and reality\footnote{Notice that any demand that arrived after $\tau$ cannot have frozen yet, otherwise a PIO would have been placed by time $\overline{\tau}$ contradicting the assumption that the next order is at $\tau_2\geq \tau$ }.  This implies that all the demands that froze (in both simulation and reality) on the interval $(\tau,\tau_2)$ were already served. Therefore, at time $\overline{\tau}$, all dual variables corresponding to demands that had arrived by time $\tau$ take exactly the same value in both reality and the simulation at time $\tau$. 

This implies that all constraints $S_s$ corresponding to timesteps $s<\tau$ are tight in reality if and only if they were tight at time $\overline{\tau}$ in the simulation initiated at time $\tau$. Therefore, if the dual variable corresponding to $d$ became part of a tight constraint in exactly one of reality and simulation, it must be tight due to a constraint associated with a time $s'>\tau$. However, by  \cref{no_early_termination_simulation}, no such constraint could have frozen in simulation.

If such a constraint became tight in reality, then that would imply that the dual variables contributing to this constraint have increased by at least $K_0$ after time $s'>\tau$. Since the simulation ends by the time the dual variables which arrived before $\tau$ have increased by at most $K_0$, there is a demand $d'$ that had not yet arrived by time $\tau$ which is contributing to this tight constraint. However, that would mean that $d'$ would cause a phase-initiating order somewhere on the interval $(\tau,\overline{\tau}]$ contradicting the assumption that the next order is at time $\tau_2\geq \overline{\tau} >\tau$.
\end{proof}

\begin{lemma}
For any order at a time $\tau_2$, if $\tau_1$ is the time at which the order preceding $\tau_2$ was placed (or $\tau_1 = 1$, if $\tau_2$ is the first replenishment order) then $\Delta_{(\tau_1,\tau_2]} \geq K_0$.
\label{all_general_bounded_online}
\end{lemma}

\begin{proof}
If $\tau_2$ is phase-initiating, then this is a direct consequence of \cref{pio_cost_interval}. If $\tau_2$ is non-phase initiating, then this is a direct consequence of \cref{NPIO_general_bounded}.
\end{proof}

\begin{lemma}
Any non-item-phase initiating item order at some time $\tau_2$ of item type $i$ must be triggered by a demand $d$ which was serviceable at the previous item order $\tau$ of item type $i$. Moreover, for any demand $d_2$ of item type $i$ served prematurely at time $\tau$, when the wavefront value is $\tau_2$, $b_{d_2} \geq H^{d_2}_\tau$.
\label{add_on_is_good}
\end{lemma}

\begin{proof}
By definition of non-item phase initiating item orders, the interval $(s, \tau_2]$ associated with the constraint $S_s$ that induced the item order must contain $\tau$, otherwise it would be item-phase initiating. Therefore, $s<\tau$ and $H^d_{\tau} \leq b_d$.

Now suppose that for some demand $d_2$ of item type $i$ which was served for a holding cost, $b_{d_2} < H^{d_2}_{\tau}$. Recall that by choice of the order in which we augmented the order with demands, $b_{d_2}$ must have already frozen before $\tau_2$. That means that there is some constraint $S_{s'}$ associated with a set $S$ of demands and a timestep $s' \in (\tau,\tau_2)$ that caused $b_{d_2}$ to freeze. Moreover, $d \notin S$ otherwise $d$ would have frozen before $\tau_2$ However, if we consider $S \cup \{ d \}$, clearly the RHS has only increased by $H^d_{s'}<H^d_{\tau} \leq H^d_{\tau_2}$ implying that $d$ must have frozen at $\tau_2$ contradicting the assumption that it induced an item order.
\end{proof}

\begin{lemma}
For any order of an item type $i$ at some time $\tau_2$, $\Delta^{D_i}_{(\tau_1,\tau_2],sim} \geq K_i$ where $\tau_1$ is the time at which the previous order of item $i$ took place, or if no such previous time exists, $\tau_1=1$.
\label{jrp_item_ordering_online_simple}
\end{lemma}

\begin{proof}
If the item order of item type $i$ at $\tau_2$ is item-phase initiating, and triggered by a constraint at some time $s \geq \tau_1$, then the total increase in $\sum_{(i,t) \in D_i} z^{it}_{is}$ from $\tau_1$ till the moment the order was triggered (possible during the simulation initiated at time $\tau_2$ is at least $K_i$.

If the item order of item $i$ at $\tau_2$ is non-item-phase-initiating, then by \cref{add_on_is_good}, when the order of item type $i$ at $\tau_2$ is triggered (possibly in simulation), for every demand $(i,t)$ served prematurely at $\tau_1$, $b_{it} \geq H^{it}_{\tau_1}$. Therefore, for all $(i,t) \in F$, $\sum_{(i,t) \in F} b_{it}+b_d \geq H^d_{\tau_1}+\sum_{(i,t) \in F} H^{it}_{\tau_1}>K_i$. Moreover, we know that at time $\tau_1$, because all demands in $F \cup \{d\}$ had desired service times after $\tau_1$, all their $b$ variables were $0$. Therefore, the increase in their sum is at least $K_i$ as desired.
\end{proof}

\begin{lemma}
The total non-regular item ordering cost is at most $2B$.
\label{jrp_item_ordering_cost_factor_2_online}
\end{lemma}

\begin{proof}
By \cref{jrp_item_ordering_online_simple}, we know that the cost of any non-regular item order of item type $i$ at some time $\tau$ is at most $\Delta^{D_i}_{(\tau_p,\tau]}+\Delta^{D_i}_{\tau_2,sim}$ where $\tau_p$ is when the item was previously ordered before $\tau$ (or time $1$ if there was no previous order). Clearly if $O^N_i$ is the set of non-regular demands of item type $i$, $\sum_{i \in [N]}\sum_{(i,\tau) \in O^N_i} \Delta^{D_i}_{(\tau_p,\tau]} \leq \sum_{i \in [N]} \sum_{d \in D_i} b_d = B$. Moreover, because the total increase of $b$ variables in simulation is at most $K_0$, $\sum_{i \in [N]}\sum_{(i,\tau) \in O^N_i} \Delta^{D_i}_{\tau,sim} \leq \sum_{\tau \in O} K_0 \leq \sum_{\tau \in O} \Delta^D_{\tau_{prev},\tau}$ where $\tau_{prev}$ is the replenishment order preceding $\tau$ (or $\tau_{prev}=1$ if there is no preceding order). The last inequality is by \cref{all_general_bounded_online}. However, by definition, $\sum_{\tau \in O} \Delta^D_{\tau_{prev},\tau}=\sum_{d \in D} b_d=B$.
\end{proof}

By combining \cref{NPIO_general_bounded} and \cref{pio_cost_interval} we obtain that if the dual solution has value $B$, the general ordering cost is at most $B$. By \cref{jrp_item_ordering_cost_factor_2_online}, we know that the item ordering cost is at most $2B$. It is also clear that every demand $d$ had delay cost at most $b_d$. Since demands can only be served for a holding cost because they froze in the simulation or during the identification of which requests to serve prematurely, we can bound the total holding cost by the total general+item ordering cost $3B$ giving us:

\begin{lemma}
{\sf{ALGJRP-{simple}}} is at most $7-$competitive.
\label{online_simpler_competitive}
\end{lemma}

\subsection{5-competitive analysis}
\label{sec:jrp-analysis}

The previous analysis was loose in several ways. For example, the delay cost is only really the total increase of $b$ variables while their corresponding demands are active. Moreover, the non-regular item ordering cost is really only $\sum_{(i,t) \in D-R} (b_{it}+b^{semi}_{it})+(P-1)K_0$ where $R$ is the set of regular demands, $P$ is the number of phases and $b^{semi}_{it}$ is the total increase in $b_{it}$ after $b_{it}$ became semi-active. These facts will allow us to ensure that the sum of delay costs and non-regular item ordering costs is at most $2B$ instead of $3B$. We will also be able to tighten the bound on holding costs from $3B$ to $2B$ bringing us down to the stated $5$-competitive algorithm.

We will refer to holding cost incurred due to demands freezing in simulation as non-add-on holding cost and holding cost incurred due to demands served during the prematurely serving additional requests step as add-on holding cost. We will also let $B^{active}$ and $B^{semi}$ be the total increase of $b$ variables while active (respectively semi-active). Similarly, we will let $\Delta^{S,semi}_{(\tau,\tau_2]}$ refer to the total increase in $b$ variables corresponding to demands in $S$, while semi-active, as the wavefront value increased from $\tau$ to $\tau_2$.

\textbf{Modification to the algorithm:} In the premature service step for item $i$, we will not always add demands onto the order in ascending order until the cost would exceed $K_i$. If item $i$ was ordered because of a demand freezing in simulation, we instead add demands onto the order until the cost would exceed $K_i-\alpha$ where $\alpha$ is the total increase of $b$ variables of demands of item $i$ that froze during the simulation. We will call this algorithm {\sf{ALGJRP-{final}}}.

To bound the item ordering cost earlier, we wanted to be able to guarantee that for a pair of consecutive orders of item $i$ at $\tau_1,\tau_2$, the sum of increase in $b$ variables corresponding to $i$ in the simulation at $\tau_2$ and between $\tau_1$ and $\tau_2$ was at least $K_i$. To do this, we used the fact that the premature service step served demands with a holding cost of $K_i$. However, if during the simulation, $b$ variables corresponding to demands of item $i$ that froze in simulation increased by $\alpha<K_i$, then it is sufficient to add on demands until the holding cost exceeds $K_i-\alpha$. Crucially, this will enable us to bound the add-on holding cost using the total increase in dual variables of item $i$ since the previous item order, enabling us to say that the add-on holding cost is at most $B$ instead of $2B$.

Now we state the algorithm {\sf{ALGJRP-{final}}} and a couple of subroutines it relies on more formally.

\begin{algorithm}[H]
\caption{{\sf{ALGJRP-{final}}}}
\begin{algorithmic}[1]

\State Set $\tau \gets 1$.
\State Set $b_{it} \gets 0$ for all demands $(i,t)$ that have arrived at time $1$.
\State Set $z^{it}_{is} \gets 0$ for all demands $(i,t)$ and timesteps $s$.
\State Set $z^{it}_{s} \gets 0$ for all demands $(i,t)$ and timesteps $s$.
\State Let $A$ be the set of demands that have arrived at time $1$.
\State Let $U \gets \emptyset$, $X \gets \emptyset$.

\While{$\tau < T$}

    \State Let $(i,t)$ be the demand such that $H^{it}_{\tau} \neq H^{it}_{\tau+1}$.

    \If{$\tau \ge t$ and $(i,t) \in A \cup X$}
        \State Increase $b_{it}$ gradually until $H^{it}_{\tau+1}$ is reached, or there exists a timestep $s$ such that:
        \[
        b_{it} \ge H^{it}_s, \quad
        \sum_{(i,t) \in D_i} z^{it}_{is} = K_i, \quad
        \sum_{(i,t) \in D} z^{it}_s = K_0.
        \]

        As $b_{it}$ increases, for all timesteps $s$ such that $b_{it} \ge H^{it}_s$ and $\sum_{(i,t)\in D_i} z^{it}_{is} < K_i$, simultaneously increase $z^{it}_{is}$ at the same rate as $b_{it}$. Moreover, for timesteps $s$ such that $b_{it} \ge H^{it}_s$ and $\sum_{(i,t)\in D_i} z^{it}_{is} = K_i$, simultaneously increase $z^{it}_s$ at the same rate as $b_{it}$

    \EndIf

    \If{$b_{it} < H^{it}_{\tau+1}$ and $(i,t) \in A$} :
        \State Place an order at time $\tau$.

        \State Let $S_\tau$ be the set of item types $i'$ such that:
        \[
            \sum_{(i',t') \in D_{i'}} z^{i't'}_{i's} = K_{i'}, \qquad
            \exists (i',t'') \in A: \ b_{i't''} \ge H^{i't''}_s.
        \]

        \State Set $(\boldsymbol{\alpha}, S^{\mathrm{sim}}_\tau, D^{sim}_{\tau},H)
        \gets \mathrm{Simulation}(b,z,D,\tau+1,A,U,X,H)$.

        \ForAll{$i' \in S_\tau$}
            \State Prematurely serve the set of demands $\mathrm{Premature\_service}(b,z,D,\tau,i',K_{i'})$.
            \State Move those demands from $A$ to $X$.
        \EndFor

        \ForAll{$i' \in S_\tau \setminus S^{\mathrm{sim}}_\tau$}
            \State Prematurely serve the demands $\mathrm{Premature\_service}(b,z,D,\tau,i',K_{i'}-\boldsymbol{\alpha}_{i'})$.
            \State Move those demands from $A$ to $X$.
        \EndFor

        \State $A \gets A \setminus \{(i,t)\}$.
        \State $X \gets X \setminus \{(i,t)\}$.
        \State $U \gets U \cup \{(i,t)\}$.
    \EndIf
    \State $\tau \gets \tau+1$
    \State Add all newly arrived demands to $A$.
\EndWhile

\end{algorithmic}
\end{algorithm}

\begin{algorithm}[H]
\caption{Simulation(b,z,D,$\tau_{start}$,A,\textit{U},X,H)}
\begin{algorithmic}[1]
\State Set $\tau_{sim} \gets \tau_{start}$ 
\State Set $\Delta = 0$
\State Set $\alpha_i \gets 0$ for all item types $i$
\State Set $S^{sim}_{\tau} \gets \emptyset$ and $D^{sim}_{\tau} \gets \emptyset$
\While{$\Delta<K_0$ and $A \cup X \neq \emptyset$}

    \State Let $(i,t)$ be the demand such that $H^{it}_{\tau_{sim}} \neq H^{it}_{\tau_{sim}+1}$.

    \If{$\tau_{sim} \ge t$ and $(i,t) \in A \cup X$}
        \State Increase $b_{it}$ gradually until $H^{it}_{\tau_{sim}+1}$ is reached, or there exists a timestep $s$ such that:
        \[
        b_{it} \ge H^{it}_s, \quad
        \sum_{(i,t) \in D_i} z^{it}_{is} = K_i, \quad
        \sum_{(i,t) \in D} z^{it}_s = K_0.
        \]

        As $b_{it}$ increases, for all timesteps $s$ such that $b_{it} \ge H^{it}_s$ and $\sum_{(i,t)\in D_i} z^{it}_{is} < K_i$, simultaneously increase $z^{it}_{is}$ at the same rate as $b_{it}$. Moreover, for timesteps $s$ such that $b_{it} \ge H^{it}_s$ and $\sum_{(i,t)\in D_i} z^{it}_{is} = K_i$, simultaneously increase $z^{it}_s$ at the same rate as $b_{it}$
    \EndIf
    \State $\Delta \gets \Delta + b_{it}-H^{it}_{\tau_{sim}}$ 
    \State $\alpha_i \gets \alpha_i + b_{it}-H^{it}_{\tau_{sim}}$
    \If{$b_{it} < H^{it}_{\tau+1}$ and $(i,t) \in A \cup X$}:
        \State $A \gets A \setminus \{(i,t)\}$.
        \State $X \gets X \setminus \{(i,t)\}$.
        \State $U \gets U \cup \{(i,t)\}$.
        \State For all $s > \tau_{sim}$, $H^{it}_s \gets H^{it}_{\tau_{sim}}$ \If{$(i,t) \in A$}:
            \State $S^{sim}_{\tau} \gets S^{sim}_{\tau} \cup \{i\}$
            \State $D^{sim}_{\tau} \gets D^{sim}_{\tau} \cup \{(i,t)\}$

        \EndIf
    \EndIf

\EndWhile
\State \Return $(\alpha,S^{sim}_{\tau},D^{sim}_{\tau},H)$

\end{algorithmic}
\end{algorithm}

\begin{algorithm}[H]
\caption{$\mathrm{Premature\_service}(b,z,D,\tau,i', threshold$)}
\begin{algorithmic}[1]
\State For each demand $(i',t')$ such that $t' > \tau$, let $g_{i't'}$ be the earliest time $s>t'$ such that $H^{i't'}_s \geq H^{i't'}_\tau$.
\State Sort the demands of item type $i'$ with desired service time $t'>\tau$ in ascending order of their corresponding $g_{i't'}$. Let $d_k$ be the $k$th demand in this sequence. Let $L$ be the length of this sequence
\State $\beta \gets 0, D_{premature} \gets \emptyset$
\State $k \gets 1$
\While{$k \le L$}
    \If{$\beta + H^{d_k}_{\tau} > \text{threshold}$}
        \State \Return $D_{premature}$
    \Else
        \State $D_{premature} \gets D_{premature} \cup \{d_k\}$
        \State $\beta \gets \beta + H^{d_k}_{\tau}$
        \State $k \gets k + 1$
    \EndIf
\EndWhile
\State \Return $D_{premature}$
\end{algorithmic}
\end{algorithm}

\textbf{The modification and results from the $7$-competitive analysis:} Most of the results from the previous section did not use the fact that the threshold was $K_i$. The only exceptions are \cref{jrp_item_ordering_online_simple} and the two lemmas that were derived using it, \cref{jrp_item_ordering_cost_factor_2_online} and \cref{online_simpler_competitive}.  

Now, we begin by proving a useful consequence of the clipping step of the algorithm:

\begin{lemma}
If a demand $d$ freezes at time $\tau_{sf}$ in a simulation at time $\tau$, but then actually freezes at a time $\tau_{f}$, then either $\tau_f \geq \tau_{sf}$ or the demand froze due to a constraint $S_s$ with $s>\tau$.
\label{incorrect_is_ok}
\end{lemma}

\begin{proof}
Let $\tau_e \geq \tau_{sf}$ be the time at which the simulation ended. Notice that at every time $\tau' \in [\tau,\tau_{e}]$, the actual value of a dual variable $b_d$ is at most its value at time $\tau'$ in the simulation. This is because the delay cost function of every demand was clipped to be no greater than the value at which it froze in the simulation.

This implies that, for every timestep $s \leq \tau$, $z^{it}_{is}$ and $z^{it}_s$ are not larger when the wavefront is actually at $\tau'$ than when the wavefront was at $\tau'$ in simulation. Therefore, if a dual variable froze at time $\tau_{sf} \in [\tau,\tau_e]$ in the simulation at time $\tau$, then the only way it could have frozen earlier in reality would be if it froze because of constraints $S_s$ with $s>\tau$.
\end{proof}

Now we prove a slightly stronger version of \cref{jrp_item_ordering_online_simple}. Recall that we cannot directly use \cref{jrp_item_ordering_online_simple} because we changed the threshold for the premature service step.

\begin{lemma}
Suppose an item order of item type $i$ was placed at time $\tau$. Let $F$ be the set of demands of item type $i$ that froze during the simulation at $\tau$, because of the constraints that triggered this item order of item type $i$. Then $\Delta^{D_i}_{(\tau_{old},\tau]} + \Delta^F_{\tau,sim} \geq K_i$ if $\tau_{old}$ is the time at which the previous item order of item $i$ occurred. If no such previous time exists $\tau_{old}=1$.
\label{item_orders_well_spaced}
\end{lemma}

\begin{proof}
If the item order of item type $i$ at time $\tau$ is item-phase-initiating (i.e. took place because of the constraints $S_s$ with $s>\tau_{old}$), then when the wavefront was at $\tau_{old}$, we have $z^{d}_{is}=z^d_s=0$ for all demands $d$. If the item order was triggered before the simulation at wavefront value $\tau$ began, then clearly at wavefront value $\tau$, $\sum_{(i,t) \in D_i} z^{it}_{is}=K_i$, implying that $\Delta^{D_i}_{(\tau_{old},\tau]} \geq K_i$. If instead it was triggered during the simulation at wavefront value $\tau_2>\tau$ during the simulation at time $\tau$, define $\beta$ as the value of $\sum_{(i,t) \in D_i} z^{it}_{is}$ when the wavefront value was $\tau$ just before the simulation at time $\tau$. We know that $\Delta^{D_i}_{(\tau_{old},\tau]} \geq \beta$. Let $U$ be the set of demands such that $z^{it}_{is}$ increased during the simulation initiated at time $\tau$. We know that at some point during the simulation $S_s$ become tight and in particular that $\sum_{(i,t) \in D_i} z^{it}_{is}=K_i$. This implies that every demand in $U$ must have frozen because of $S_s$ during the simulation. Therefore, $U \subset F$. Since $\sum_{(i,t) \in U} z^{it}_{is}$ increased by at least $K_i-\beta$, $\Delta^U_{\tau,sim} \geq K_i - \beta \geq K_i - \Delta^{D_i}_{(\tau_{old},\tau]}$. Since $U \subset F$, we have the desired result in this case.

If the item order at $\tau$ is a non-item-phase-initiating item order triggered by a demand $d$, and the simulation at $\tau_{old}$ was $i-correct$, and the add-on holding cost threshold was $K_i-\alpha$, then the actual increase in dual variables corresponding to demands of item type $i$ that froze in simulation at $\tau_{old}$ is at least $\alpha$. Moreover, by  \cref{add_on_is_good} the sum of increase in demands that were served by add-on holding cost (including in simulation at $\tau$) and $b_d \geq H^d_{\tau_{old}}$ must be at least $K_i-\alpha$ giving us the desired result.

If the item-order is non-item-phase initiating and the simulation associated with $\tau_{old}$ was not $i-$correct, then we still have that by \cref{add_on_is_good} the sum of increase of budgets in demands that were served by add-on holding cost and $b_d \geq H^d_{\tau_{old}}$ must be at least $K_i-\alpha$ (where $\alpha$ is defined similarly to the previous case). Suppose the increase in dual variables corresponding to demands of item type $i$ that froze in simulation were smaller than $\alpha$. Then some demand that froze in the simulation at $\tau_{old}$ froze earlier in reality. This implies by  \cref{incorrect_is_ok} \footnote{Technically, if we did not apply the clipping procedure, then Lemma \ref{item_orders_well_spaced} would still hold, however, the proof becomes much more complex without \cref{incorrect_is_ok}} that this demand must have frozen because of a constraint associated with time $\tau'>\tau_{old}$. This in turn would imply that the increase in dual variables of item type $i$ after $\tau'$ is at least $K_i$ implying that the total increase after $\tau_{old}<\tau'$ must also be at least $K_i$.
\end{proof}

Using the strengthened result about the change in $\sum_{(i,t) \in D_i} b_{it}$ between consecutive item orders of item $i$, we will be able to bound the add-on holding cost.

\begin{lemma}
The total add-on holding costs associated with an item order of item type $i$ at time $\tau$ are at most $\Delta^{D_i}_{(\tau_{old},\tau]}$ where $\tau_{old}$ is when the previous item order of item $i$ occurred. If there is no such previous item order, $\tau_{old}=1$.
\label{add_on_holding_is_bounded}
\end{lemma}

\begin{proof}
By \cref{item_orders_well_spaced}, we know that $\Delta^{D_i}_{(\tau_{old},\tau]}+ \Delta^F_{\tau,sim} \geq K_i$. We know by construction that the threshold for add-on holding cost in the premature service step for this order of item $i$ at time $\tau$ is $K_i-\Delta^{D_i}_{\tau,sim} \leq K_i - \Delta^F_{\tau,sim} \leq \Delta^{D_i}_{(\tau_{old},\tau]}$ as desired.
\end{proof}

In order to bound the total holding cost, we also need to bound the holding cost of demands served due to the simulation (the non-add-on holding cost), which is what the following lemma does.

\begin{lemma}
The total non-add-on holding cost for a non-phase initiating order at time $\tau$ preceded by an order $\tau_{prev}$ is at most $\Delta_{(\tau_{prev},\tau]}$. For a phase-initiating order at time $\tau$, let $R_\tau$ be the set of item types $i$ such that a regular item order of type $i$ occurred at $\tau$. Then the non-add-on holding cost at time $\tau$ is at most $\sum_{i \in R_{\tau}} \Delta^{D_i}_{(\tau_{prev},\tau]} + \sum_{i \notin R_{\tau}} \Delta^{D_i,semi}_{(\tau_{prev},\tau]}$ (where $\tau_{prev}$ is the order preceding $\tau$ or $1$ if $\tau$ is the first order).
\label{simulation_is_bounded}
\end{lemma}

\begin{proof}
Non-add-on holding cost is incurred for a demand $d$ with desired service time $\tau_2$ if $b_d$ freezes in a simulation initiated at time $\tau<\tau_2$. In this case, by \cref{simulation_tight_only_past}, at the start of the simulation, $b_d=0$ and at the end $b_d \geq H^d_{\tau}$. Therefore, the total non-add on holding cost incurred by any replenishment order is at most $K_0$, its general ordering cost. We know from \cref{NPIO_general_bounded} that for any non-phase initiating order at a time $\tau$ preceded by an order at time $\tau_{prev}$, $\Delta_{(\tau_{prev},\tau]} \geq K_0$. 

By  \cref{pio_cost_interval} we know that for a phase-initiating order at time $\tau$ with $\tau_{prev}$ defined as above an even stronger result holds. Let the constraints that triggered the order at time $\tau$ are $S_s$ ($s>\tau_{prev}$ because this is a phase initiating order) and $S$ is the set of item types $i$ such that $\sum_{(i,t) \in D_i} z^{it}_{is}=K_i$ when the wavefront was at $\tau$. Then for every item type $i \in S - R_{\tau}$, we know that, for all active demands $d \in D_i$, $z^{d}_s=z^{d}_{is}=0$ (otherwise a regular item order of type $i$ would have been placed). Therefore,  $$\sum_{i \in R_{\tau}} \Delta^{D_i}_{(\tau_{prev},\tau]} + \sum_{i \notin R_{\tau}} \Delta^{D_i,semi}_{(\tau_{prev},\tau]} \geq \sum_{i \in R_{\tau}} \Delta^{D_i}_{(\tau_{prev},\tau]} + \sum_{i \in S- R_{\tau}} \Delta^{D_i,semi}_{(\tau_{prev},\tau]} \geq K_0 + \sum_{i \in S} K_i$$ 
(where the last inequality is due to \cref{pio_cost_interval})and because the non-add-on holding costs for the order at $\tau$ are at most $K_0$, we have the desired result.
\end{proof}

Now we can derive a stronger bound on the non-regular item ordering cost.

\begin{lemma}
The non-regular item ordering cost of {\sf{ALGJRP-{final}}} is at most $B+B^{semi}-\sum_{\tau \in \mathcal{P}} \sum_{i \in R_{\tau}} \Delta^{D_i,semi}_{(\tau_{prev},\tau]}$ where $\mathcal{P}$ is the set of phase initiating orders and $\tau_{prev}$ is the order preceding $\tau$ (or $1$ if $\tau$ is the first order).
\end{lemma}

\begin{proof}
By \cref{item_orders_well_spaced}, we know that the cost of any non-regular item order of item type $i$ at a time $\tau$ (with $\tau_{old}$ defined to be when the previous item order of type $i$ occurred, or if no previous item order of type $i$ exists, $1$) is at most $\Delta^{D_i}_{(\tau_{old}.\tau]}+\Delta^F_{\tau,sim}$ where $F$ is the set of demands of item type $i$ that froze in the simulation at $\tau$. We will charge $\Delta^{D_i}_{(\tau_{old},\tau]}$ to the increase in $b$ variables of item type $i$ between $\tau_{old}$ and $\tau$. If the replenishment order (not necessarily of item $i$) succeeding $\tau$, $\tau_{next}$ is not phase initiating, then by \cref{pio_incorrect_simulation}, we can charge $\Delta^{F}_{\tau,sim}$ to the semi-active increase in dual variables of item type $i$ on the interval $(\tau,\tau_{next}]$.

So far, for any non-phase initiating order $\tau$ and the order that preceded it $\tau_{prev}$ ($\tau_{prev}=1$ if $\tau$ is the first order), we have charged the total increase in $b$ variables as the algorithm moved from $\tau_{prev}$ to $\tau$  at most once and the semi-active increase of $b$ variables at most a second time if $\tau$ is not a phase-initiating order. If $\tau$ is a phase-initiating order, then the increase in dual variables of item type $i$ on the interval $(\tau_{prev},\tau]$ has been charged once if the first item order of type $i$ at $\tau$ or later is non-regular and $0$ times if it is regular.

If we let $A$ be the set of orders succeeded by a phase initiating order, then the total item ordering cost left to charge is at most $\sum_{\tau_{prev} \in A} \Delta^{F_{\tau_{prev}}}_{\tau_{prev},sim} \leq K_0 |A|$ where $F_{\tau_{prev}}$ is the set of demands that froze in a simulation at $\tau_{prev}$. Let $\tau$ be the phase initiating order succeeding $\tau_{prev}$. By the same arguments as \cref{simulation_is_bounded}, we know that for each such $\tau$, $\sum_{i \in R_{\tau}} \Delta^{D_i}_{(\tau_{prev},\tau]}+\sum_{i \notin R_{\tau}} \Delta^{D_i,semi}_{(\tau_{prev},\tau]} \geq K_0$. Therefore, we can charge $\Delta^{F_{\tau_{prev}}}_{\tau_{prev},sim}$ to the increase in $b$ variables of item types in $R_{\tau}$ on the interval $(\tau_{prev},\tau]$ and the semi-active increase of $b$ variables of item types not in $R_\tau$. 

At this point, for every phase initiating order $\tau$, the total increase in $b$ variables of item types in $R_{\tau}$ on the interval $(\tau_{prev},\tau]$ has been charged at most once while for item types not in $R_{\tau}$, the total active increase in corresponding $b$ variables on the same interval has been charged at most once and the total associated semi-active increase has been charged twice. 
\end{proof}

Because every demand is served before it becomes semi-active or inactive, we can actually derive a sharper bound on the delay costs of the algorithm.

\begin{lemma}
The delay cost of {\sf{ALGJRP-{final}}} is at most $B^{active}$.
\label{delay_is_bounded}
\end{lemma}

\begin{proof}
A demand becomes semi-active or inactive the moment it's served. While a demand is active, the dual variable is always equal to the delay cost incurred by the demand. Therefore we have the desired result.
\end{proof}

\begin{lemma}
{\sf{ALGJRP-{final}}} is 5-competitive.
\end{lemma}

\begin{proof}
By \cref{pio_cost_interval} and  \cref{NPIO_general_bounded}, the total general ordering cost and the regular item ordering cost is at most $B$.
By \cref{delay_is_bounded}, the delay costs are at most $B^{active}$. By  \cref{add_on_holding_is_bounded}, the add-on holding costs are at most $B$. We know by \cref{simulation_is_bounded} that the non-add-on holding cost is at most $B$. The non-regular item ordering cost by using  \cref{item_orders_well_spaced} in conjunction with  \cref{simulation_is_bounded} is at most $B+B^{semi}$. Adding all this together, we get a $5$-competitive algorithm.
\end{proof}

\section{Improving the single-item lot-sizing competitive ratio to $1+\phi$} \label{sec:phi}

We obtain an improved competitive ratio for the single-item lot-sizing problem, by a simple modification to our algorithm: instead of including additional items provided that the holding cost incurred is under $K$, we will add items until the holding cost reaches $(\phi-1)K \approx 0.618K$.

We will need to distinguish between two kinds of orders. Let $s'$ be the constraint that triggered a replenishment order at timestep $s$. We call the replenishment order at timestep $s$ phase-initiating if there were no replenishment orders placed during the interval $(s',s)$. For any subset of the demands $S$, let $\Delta^S_{(a,b]}$ be the change in $\sum_{t \in S} b_t$ from time $a$ to time $b$. 

\textbf{Phase-initiating orders (PIOs):} Between a phase initiating order at time $s$ triggered by a constraint that became tight at time $s'$ and the order immediately before it at time $s_{prev}$ (let $s_{prev}=0$ if no such order exists), we can conclude that $\Delta_(s_{prev},s] \geq K$ because when the algorithm was at time $s_{prev}$, $\sum_{t} z^t_{s'}=0$. Any demand with desired service time $t \leq s_{prev}$ would have been served by a previous order. Therefore, the total delay cost of demands served by this order is at most $\sum_{t \in (s_{prev},s]} b_t \leq \Delta_{(s_{prev},s]}$. The total holding cost associated with this order is, for some $\alpha \geq 0$ $(0.618-\alpha_s)K \leq (0.618-\alpha_s) \Delta_{(s_{prev},s]}$ and the total ordering cost is at most $K \leq \Delta_{(s_{prev},s]}$. Therefore the total ordering, delay and holding cost associated with the order at time $s$ is at most $(2.618-\alpha_s)\Delta_{(s_{prev},s]}$.

\textbf{Non-phase-initiating orders (NPIOs):} Let $s$ be a non-phase-initiating order triggered by demand $t$ and constraint $s_{trig}$ let $s_{prev}$ be the order that immediately preceded $s$. Let $S_{prev}$ be the set of demands that were served at time $s_{prev}$ for a holding cost. Using arguments similar to  \cref{ordering_costs_single_item_3}, we can conclude that $\Delta^{S_{prev}}_{(s_{prev},s]} \geq (0.618-\alpha_{s_{prev}})K$ and $b_t \geq \alpha_s$ since this order has been triggered by $s_{trig} \leq s_{prev}$. Let $S^{delay}$ be the set of demands served at $s$ for a delay cost except $t$ and $S^{hold}$ the set of demands served at $s$ for holding cost. We know that $\Delta_{(s_{prev},s]} > 0.618K + \Delta^{S^{delay}}_{(s_{prev},s]}+ b_t - \alpha_{s_{prev}}K $.

Therefore, the ordering cost associated with time $s$ is at most $\phi ( \Delta_{(s_{prev},s]} - \Delta^{S^{delay}}_{(s_{prev},s]}  - b_t + \alpha_{s_{prev}}K)$. The holding cost associated with time $s$ is at most $( \Delta_{(s_{prev},s]} - \Delta^{S^{delay}}_{(s_{prev},s]}  - b_t + \alpha_{s_{prev}}K - \alpha_s K)$. Since the delay cost for demands other than $t$ is at most $\Delta^{S^{delay}}$ and the delay cost for $t$ is $b_t$, we get that the total cost corresponding to demands served in this order is $$(\phi+1)(\Delta_{(s_{prev},s]}) +(\phi+1)\alpha_{s_{prev}}K-\phi b_t-\alpha_sK \leq (\phi+1)(\Delta_{(s_{prev},s]}) + \alpha_{s_{prev}}K-\alpha_s K.$$

When summing over a PIO and the NPIOs that follow it, the last two terms telescope, leaving us with a total cost of at most $2.618\OPT$.

\section{Non-monotone holding-delay costs costs}

We will now show that even the single-item lot-sizing problem with non-monotone holding-delay costs generalizes set cover. Note that this shows hardness for even the offline version of the problem. To our knowledge, this straightforward observation has not appeared in prior work on the JRP.

\begin{lemma}
An $\alpha$-approximation algorithm for single-item lot-sizing with holding-delay costs implies an $\alpha$-approximation for set cover.
\end{lemma}

\begin{proof}
Given an instance of the set cover problem with a collection of elements $U=u_1,u_2...u_n$ and family of subsets of $U$ $\mathcal{S} = S_1,S_2...S_m$. We construct an instance with $m+n$ timesteps. The $i$th element has desired service time $m+i$ and, for $k \leq m$, $H^{m+i}_{k} = 0$ if $u_i \in S_k$ and $H^{m+i}_k=\infty$ if $u_i \notin S_k$. For $k > m$, $H^{m+i}_{k}=0$ if $k=m+i$ and $H^{m+i}_{k}=\infty$ otherwise. The ordering cost is always $1$.

Clearly any solution with finite cost serves all demands for a holding-delay cost of $0$. Therefore, if a solution of finite cost opens $C$ orders, it has objective value $C$. Given a solution that opens $C$ orders, we can construct a solution that opens $C$ orders using only the first $m$ timesteps as follows. For each timestep $t$ from $m+1$ to $m+n$, if the single-item solution opens an order at time $t$, let $f(t)$ be the smallest timestep such that $H^{t}_{f(t)}=0$. If the set cover instance is feasible, then there is some set containing the element corresponding to demand $t$, which would imply that $f(t) \leq m$. Close the order at time $t$ and open an order at time $f(t)$. Serve $t$ at $f(t)$ instead of at $t$. We know that $H^{t'}_{t}=\infty$ for all demands $t' \neq t$. Therefore, doing this does not leave any demand unserved. Repeating this process for all $t \in [m+1,m+n]$ gives us a solution that only places orders in the first $m$ timesteps.

Now, we find a set cover that uses only $C$ sets. Let $O$ be the set of timesteps in which orders are placed. This means that every demand $t$ had $0$ holding cost at some timestep in $O$. By construction, this means the collection $\{S_k : k \in O\}$ is a valid set cover of size $C$.

Therefore, the optimal value of the single-item lot-sizing instance is at most the optimal value of the set-cover instance.

Now we prove the reverse: Given a solution to the set cover instance that is a family $F$ of sets, we can construct a solution for the single-item lot-sizing instance which opens an order at time $k$ if $S_k \in F$ and serves every demand $u_i$ at the smallest $k$ such that $u_i \in S_k$. Since the set-cover solution $F$ covered every element, the single-item lot-sizing solution serves every demand and has the same cost as the set cover solution.

Therefore, the optimal value of the set-cover instance is at most the optimal value of the single-item lot-sizing instance. Since we already proved the reverse of this earlier, we know that the optimal value of both instances is the same. Moreover, we saw that any solution of cost $C$ for single-item lot-sizing can be converted into a solution of no-greater cost for the set cover instance. This implies that an $\alpha-$approximation algorithm for single item lot-sizing with non-monotone holding-delay costs can be used to obtain an $\alpha$-approximation for set cover. 
\end{proof}
\bibliographystyle{siamplain}
\bibliography{references}
\end{document}